\documentclass[conference]{IEEEtran}
\IEEEoverridecommandlockouts
\usepackage{amsmath,amsthm,amssymb,amsfonts}
\usepackage{algorithmic}
\usepackage{graphicx}
\usepackage{textcomp}
\usepackage{xcolor}
\usepackage[numbers]{natbib}
\usepackage{booktabs} 
\usepackage[utf8]{inputenc} 
\usepackage[T1]{fontenc}    
\usepackage{hyperref}       
\usepackage{url}            
\usepackage{booktabs}       

\DeclareMathOperator*{\argmax}{arg\,max}

\usepackage{braket}
\usepackage{amsfonts}       
\usepackage{nicefrac}       
\usepackage{microtype}      
\usepackage{lipsum}		
\usepackage{graphicx}
\usepackage{doi}
\usepackage{multirow}
\usepackage{subcaption}
\usepackage{wrapfig}
\usepackage{bbm}
\usepackage{tikz}
\usepackage{quantikz}
\usepackage{natbib}
\usepackage{amssymb}
\usepackage{siunitx}
\usepackage{adjustbox}
\usepackage{placeins}

\usepackage{caption}
\usepackage{subcaption}
\usepackage{fancyhdr}
\usepackage{comment}

\newtheorem{theorem}{Theorem}
\newtheorem{definition}{Definition}
\newtheorem{corollary}{Corollary}

\newcommand{\je}[0]{\mathcal{J}_\mathcal{E}}
\newcommand{\ji}[0]{\mathcal{J}_\mathcal{I}}

\def\BibTeX{{\rm B\kern-.05em{\sc i\kern-.025em b}\kern-.08em
    T\kern-.1667em\lower.7ex\hbox{E}\kern-.125emX}}
\begin{document}

\title{Constructing Optimal Noise Channels for Enhanced Robustness in Quantum Machine Learning
\thanks{The project/research is supported by the Bavarian Ministry of Economic Affairs, Regional Development and Energy with funds from the Hightech Agenda Bayern.}
}

\fancypagestyle{specialfooter}{%
  \fancyhf{}
  \renewcommand\headrulewidth{0pt}
  \fancyfoot[R]{ \noindent\fbox{%
    \parbox{\textwidth}{%
        {\footnotesize \copyright 2025 IEEE. Personal use of this material is permitted. Permission from IEEE must be obtained for all other uses, in any current or future media, including reprinting/republishing this material for advertising or promotional purposes, creating new collective works, for resale or redistribution to servers or lists, or reuse of any copyrighted component of this work in other works.}
        }
    }}
}

\author{
    \IEEEauthorblockN{David Winderl\IEEEauthorrefmark{2}, Nicola Franco\IEEEauthorrefmark{2}, Jeanette Miriam Lorenz\IEEEauthorrefmark{2}}
    \IEEEauthorblockA{\IEEEauthorrefmark{2}Fraunhofer Institute for Cognitive Systems IKS, Munich, Germany
    \\\{david.winderl, nicola.franco, jeanette.miriam.lorenz\}@iks.fraunhofer.de}
}

\maketitle
\thispagestyle{plain}
\pagestyle{plain}
\thispagestyle{specialfooter} 

\begin{abstract}
Quantum Machine Learning (QML) is rapidly evolving. Recent studies suggest that quantum noise could improve classifier robustness. To verify this claim, we developed a semidefinite programming model that, given a dataset and a QML ansatz, constructs the most robust noise channel possible. In our small-scale experiments, we first explore the range of behaviors generated by this optimized noise channel. Secondly, we show that, despite recent claims, the direct impact of noise is modest compared to the substantial robustness improvements achieved through increasing the number of qubits. Aside from that, we utilize our framework to assess different noise channels in terms of their robustness and certifiability.
\end{abstract}


\section{Introduction \& Related Work}
Quantum Machine Learning (QML) emerges as a prominent example of the potential applications for Noisy Intermediate-Scale Quantum (NISQ) devices~\cite{Preskill2018quantumcomputingin}.
Research into QML is motivated by the anticipation that it could significantly improve certain computational tasks, surpassing the performance of conventional algorithms~\cite{wittek2014quantum, schuld2015introduction, biamonte2017quantum, abbas2021power, cerezo2022challenges}.
Despite these potential benefits, QML faces its own set of challenges, especially concerning susceptibility to adversarial attacks~\cite{lu2020quantum, west2023benchmarking, franco2024predominant}.
Within classical machine learning, Differential Privacy (DP) has played a crucial role in striking a balance between data utility and the need for privacy protection~\cite{dwork2014algorithmic, abadi2016deep}. 
Specifically, DP has been utilized to enhance the reliability of model predictions for specific inputs~\cite{cohen2019certified, lecuyer2019certified}. 
This principle extends naturally into Quantum Computing (QC) and introduces a novel approach to safeguard the integrity and privacy of data processed by QML models~\cite{Zhou2017}. 
\medskip    

Quantum noise channels, such as depolarizing and phase damping noise in NISQ devices, are used as sources of stochastic noise. 
This noise helps achieve DP by exploiting the natural error processes of these devices~\cite{Angrisani2023, Hirche2023, Zhou2017, Weber2021, Du2021}.
With this argumentation, \citet{Weber2021} have provided a relationship among quantum hypothesis testing and adversarial robustness.
In addition, \citet{Hirche2023} has provided a relationship between quantum DP and the quantum hockestick divergence. 
\citet{Angrisani2023} is possibly one of the most noteworthy recent publications; they have provided a more general framework for the robustness upper bound of quantum noise channels by defining the neighbourhood concept in terms of the Schatten-norm and providing a more general upper bound for noise channels in terms of a depolarizing channel as well as a single qubit Pauli channel. 
This broader framework was not included in our work as we initially focused on defining such a family of noise channels.
The linear nature of quantum channels and their relationship to semidefinite programming has been already considered, e.g. \citet{Guan_2021} used semi-definite programming to derrive the optimal bound of a quantum classifier and detect adversarial examples. 
Nonetheless, the extent to which these techniques shield quantum classifiers from practical adversarial input manipulations remains unclear.

\medskip
As indicated in our earlier work~\cite{winderl2023quantum}, the effectiveness in terms of adversarial accuracy of these approaches in providing robustness varies with the factor of depolarization noise. To make this more clear, a quantum classifier with depolarization noise strength of one is inherently robust, but is equivalent to flipping a coin. This circumstance allows asking for the \textit{best} depolarization noise strenght or more general the noise channel that is as robust as possible.
To answer the question of the most robust noise channel, we develop a family of noise channels, so called $(\alpha, \gamma)$-channels. We found that every channel in this group is naturally robust because of its origin. Next, to show the effectiveness of our strategy, we offer theoretical derivations for the bounds on robustness for \textit{depolarizing noise}~\cite{Zhou2017} and \textit{random rotations}~\cite{Huang2023} from our general framework of ($\alpha$, $\gamma$)-channels. This indicates that a broader set of channels can be described in our framework. As a final step, we add a construction of an \textit{optimal} quantum channel as a positive semidefinite constrained optimization problem.
This construction allows the experimental evaluation of a depolarizing channel against its optimally constructed counterpart, providing a more definite statement about the possible utility of quantum noise channels for DP and adversarial robustness against evasion attacks.


\section{Preliminaries}

\subsection{Quantum Adversarial Robustness}
Quantum adversarial robustness examines how specially crafted attacks can weaken a QML model, just as traditional neural networks can be fooled by adversarial inputs~\cite{goodfellow2015explaining, kurakin2017adversarial, madry2019deep}. 
While there are other means of attacking a machine learning system, here we focus on evasion attacks.
In the context of classical machine learning, adversarial robustness refers to the ability of a machine learning model to maintain a stable prediction when confronted with adversarial examples. 
Formally, let $\mathcal{Y}$ denote the set of classes and $\|\cdot\|$ denote the euclidean norm on $\mathbb{R}^n$, we assume a datapoint $x\in\mathbb{R}^n$ with a label $y \in \mathcal{Y}$. 
Given $f:\mathbb{R}^n\rightarrow [0,1]^{|\mathcal{Y}|}$ be a \textit{quantum classifier} and $\epsilon>0$, we can define an \textit{adversarial attack} as the smallest amount of perturbation $\delta$ needed to change the model's prediction, formally:
\begin{equation} 
    \operatorname{argmax}_{c\in\mathcal{Y}} f_c(x+\delta) \neq \operatorname{argmax}_{c\in\mathcal{Y}} f_c(x),
\end{equation}
where $\|\delta\| \leq \epsilon$ is bounded by a maximum $\epsilon$ budget and ${x}_{adv}=x+\delta$ denotes an adversarial example.
\medskip

In the context of QC, a study by \citet{lu2020quantum} found that Quantum Neural Networks (QNN) can fall prey to these targeted perturbations. 
As part of their work, the quantum version of the Fast Gradient Sign Method (FGSM)~\cite{goodfellow2015explaining} was introduced. On a high level FGSM computes the gradient of the loss with respect to the input and then slightly alters the input in the direction that most increases the loss, effectively creating an adversarial example by using a small pertubation.

The fact that QNNs are vulnerable to those small perturbations motivates the discussion about specific measures to counteract adversarial examples. 
To measure the effectiveness of approaches, we consider \textit{robustness accuracy} as the network's ability to main stable predictions against adversarial attacks. We define this robustness accuracy as the fraction of samples that where correctly predicted by the network under attack.


\subsection{Quantum Differential Privacy}
Differential Privacy (DP) with $(\epsilon, \delta)$ limits data exposure: a smaller $\epsilon$ boosts privacy, while $\delta$ bounds the risk of failure~\citep{dwork2014algorithmic}.
In this work, we focus entirely on $\epsilon$-DP, which implies that $\delta = 0$.~\footnote{This is justified, since typically $\epsilon$-DP is used for classification problems. Nevertheless as pointed out by \citet{Angrisani2023} $\epsilon$-DP cannot describe depolarizing noise entirely.}
Formally a mechanism M is said to be $\epsilon$-differentially private if for every pair of neighboring datasets ($D$ and $D^\prime$) and for every possible set S of outputs, the following holds:
\begin{equation}\label{eq:dp_classical}
    \text{Pr}(M(D) \in S) \leq e^\epsilon \text{Pr}(M(D^\prime) \in S),
\end{equation}
where Pr indicates \textit{probability}.
Before introducing the concept of quantum differential privacy, we want to elaborate on the concept of distance and the concept of processing in quantum computing. 
Concerning the distance, it can be defined as the trace distance between two density matrices $\rho$ and $\sigma$:
\begin{equation}\label{eq:trace_distance}
    \tau(\sigma, \rho) = \frac{1}{2} \text{Tr} \left| \sigma - \rho \right|  = \frac{1}{2} \text{Tr} \left[ \sqrt{(\sigma - \rho)^\dagger (\sigma - \rho)} \right]
\end{equation}
Computations on quantum devices are dictated by a quantum channel $\mathcal{E}$, which is a Completly Positive and Trace Preserving (CPTP) map. 
Outcomes can be characterized by a set of Positive Operator-Valued Measurements (POVMs) $\{\Pi_k\}$. 
Hence, we can express the outcome of a quantum classifer, which we assume aligning with the definition above to be a mechanism:
\begin{equation}
    \text{Pr}(M(\sigma) \in S) = y_k(\sigma) = \text{Tr} \left[\Pi_k \mathcal{E}(\sigma)\right]
\end{equation}
This leads to the definition of quantum $\epsilon$-DP on a noisy quantum algorithm $M$, by stating that for every pair of density matrices ($\sigma, \rho$) within a specific distance $\tau(\sigma, \rho)$, it holds for all possible measurement outcomes, that:
\begin{equation}\label{eq:fraction_dp}
y_k(\rho) \leq e^\epsilon y_k(\sigma), \quad \text{or} \quad \frac{y_k(\rho)}{y_k(\sigma)} \leq e^\epsilon. 
\end{equation}
\autoref{eq:fraction_dp} is of particular interest, since it indicates, that $\epsilon$ controls the change of labels for two neighbouring states.

\subsection{Bound of \texorpdfstring{$\epsilon$}{e}-DP}

Recent research on quantum DP incorporates a specific noise channel alongside the quantum channel that models the quantum classifier to establish concrete upper bounds for $\epsilon$. 
Here, we aim to outline two particular bounds relevant to our study: (i) depolarizing noise and (ii) random rotations.

\medskip
\subsubsection{Bound on Depolarizing Noise}
A depolarizing noise channel\footnote{Note that this defines a global noise channel.} is characterized by the parameter $p$, defined as follows:
\begin{equation}
    \mathcal{E}(\rho) = \frac{p}{N} I_N + (1-p) \rho.
\end{equation}
Here $N=2^n$ is the dimensionality of the density matrix and $I_N$ the identity matrix of dimensionality $N$.
\citet{Zhou2017}, initially showed that a depolarizing channel suffices $\epsilon$-DP. So given two input states $\rho$ and $\sigma$, bound by a trace distance: $\tau(\sigma, \rho) \leq \tau_D$ and a quantum classification algorithm with a POVM, perturbed by a depolarizing noise channel, the output probabilities of $\rho$ and $\sigma$ are $\epsilon$-DP: 
\begin{equation}\label{eq:dep_upperbound}
\epsilon = \ln \left( \frac{D(1-p) \tau_D}{p} + 1\right),
\end{equation}
where $D$ is the dimension of the measurement operator used in the POVM. 
\citet{Du2021} expanded on this concept by noting that the placement of the depolarizing channel within the quantum circuit does not affect the robustness guarantee. 
Additionally, they pointed out that $\epsilon$-DP can be utilized to enhance adversarial robustness.
\medskip

\subsubsection{Bound on random rotations}

\citet{Huang2023} provided a bound for $\epsilon$-DP, based on a process, where at first a layer of random rotations parametrized by $\theta_i$, with: $h_1 < \tan (\theta_i) < h_2$ are placed, further they show that for \textit{binary classification}, the following bound can be derived:
\begin{equation}
    \epsilon = \ln \left[\frac{\tau_D}{t^n} +1\right],
\end{equation}
where $t = \frac{h+1}{\sqrt[n]{2}}$ is the noise level.~\footnote{Note that to the best of our knowledge $h$ was not formally defined in \citet{Huang2023}. We concluded that it originated as $h$ from the proof in \cite[Appendix: Proof of Lemma 1]{Huang2023}. The lower bound was given as $\frac{(h+1)^n}{2}$, which we equated to $t^n = \frac{(h+1)^n}{2}$ in the definition of their $\tau_D$ and solved accordingly.} Their bound is of particular interest, since it is independent of the underlying device.

\subsection{Choi representation of noise channels}
As modern quantum devices are in the NISQ era~\cite{Preskill2018quantumcomputingin}, they are commonly affected by noise. 
This noise is often characterized using Kraus operators:
\begin{equation}\label{eq:kraus_channel}
    \mathcal{E}(\rho) = \sum_i K_i \rho K_i^\dagger \quad \text{where} \quad \sum_i K_i K_i^\dagger = I.
\end{equation}
In constrast, we utilize the Choi–Jamiołkowski isomorphism~\cite{CHOI1975285}. 
This formalism is able to represent an $n$-qubit noise channel of dimensionality $d = 2^n$ as a complex, completely positive, and trace preserving matrix:
\begin{equation}
    \je = (\mathcal{E} \otimes I_A) (\ket{\phi^+}\bra{\phi^+}) \in \mathbb{C}^{d^2 \times d^2},
\end{equation}
where $\otimes$ is a tensor product and $\ket{\phi^+}$ represents the fully entangled state and $A$ the auxiliary system which we assume to always have the same dimensionality as our \( n \)-qubit system. 
On a further note, we refer to the Choi matrix that describes the identity channel as identity Choi matrix: \( \ji = \ket{\phi^+}\bra{\phi^+} \).

\medskip

Next, it is necessary to ensure that \( \je \) is CPTP for our resulting quantum channel. 
This can be achieved by asserting that the Choi matrix is positive semidefinite (\( \je \succeq 0 \)) and that tracing out the system describing the noise channel lead to an identity matrix of the size of the subsystem (\( \text{Tr}_2 \left[\je\right] = I_A \)).
Thus, we can construct the application of a noise channel as~\footnote{See \citet{Wood_2013} for a visualization using tensor networks.}~\cite{Knee_2019}:
\begin{equation}
    \mathcal{E}(\rho) = \text{Tr}_1 \left[ (\rho \otimes I_A) \je  \right].
\end{equation}
Consequently, a POVM is constructed as follows~\cite{Knee_2019}:
\begin{equation}\label{eq:povm_choi}
    \tilde{y}_k(\rho) = \text{Tr} \left[ (\rho \otimes \Pi_k) \je  \right],
\end{equation}
where \( \Pi_k \) is a measurement operator. 

\medskip

Further, a Kraus-channel can be reconstructed from a Choi matrix by looking at its spectral decomposition: \( \mathcal{J}(\mathcal{E})  = \sum_i \lambda_i \ket{\Psi_i}\bra{\Psi_i} \)~\cite{Wood_2013}. One can then construct a matrix \( A_i \), s.t \( \lvert A_j \rangle \rangle = \sum_{i, j} A_{ij} \ket{i} \otimes \ket{j} = \ket{\Psi_i} \).~\footnote{Since we are assuming the standard basis for the construction of superoperators, this can simply be done by reshaping the matrix.}.

The individual Kraus operators can then be recovered as:
\begin{equation}
    K_i = \sqrt{\lambda_i} A_i
\end{equation}








\section{Introduction of (\texorpdfstring{$\alpha$}{alpha}, \texorpdfstring{$\gamma$}{gamma})-Channels}
In this section, we introduce a family of quantum noise channels, which allow us to derive $\epsilon$-DP bounds and robustness guarantees. 
In reviewing the proofs concerning $\epsilon$-DP conducted by \citet{Du2021}, \citet{Huang2023} and \citet{Zhou2017}, 
we observe that their bounds share similarities, particularly in their use of the property that trace preserving operations are \textit{contractive} with respect to the trace distance, more formally:
\begin{equation}
    \tau(\mathcal{E}(\sigma), \mathcal{E}(\rho)) \leq \tau(\sigma, \rho).
\end{equation}
In addition, they typically manage to establish bounds for the minimum probability of the noise channel, i.e: $y_k(\sigma) > 0$.
This reasoning inspires the definition of a family of noise channels that can be demonstrated to meet the requirements of $\epsilon$-DP.
\medskip

Building on this idea, we formalize the concept by introducing $(\alpha, \gamma)$-channels:
\begin{definition}[$(\alpha, \gamma)$-channels]\label{def:alpha_gamma_channel}
    We define an $(\alpha, \gamma)$-channel as an arbitrary quantum channel $\mathcal{E}_{(\alpha, \gamma)}(\rho): \mathbb{C}^N \rightarrow \mathbb{C}^N$, with the following properties:
    \begin{enumerate}
        \item $\forall \sigma, \rho. \quad \tau(\mathcal{E}_{(\alpha, \gamma)}(\rho), \mathcal{E}_{(\alpha, \gamma)}(\sigma)) \leq \alpha \tau(\sigma, \rho)$
        \item $\forall \Pi_k, \sigma. \quad \text{Tr} \left[\Pi_k \mathcal{E}_{(\alpha, \gamma)}(\sigma)\right] \geq \gamma$
    \end{enumerate}
    with $\alpha, \gamma \in [0;1]$.
\end{definition}
\smallskip

Intuitively, any $(\alpha, \gamma)$-channel has a contraction factor of at most $\alpha$. Any input state $\sigma$ will result in at most $\gamma$ when measured. Note that each measurement is defined by their POVM Projector $\Pi_k$.

\subsection{\texorpdfstring{$\epsilon$}{e}-DP of \texorpdfstring{$(\alpha, \gamma)$}{alpha, gamma}-channels}

As already hinted in our introduction, the definition is inspired by reviewing the process of thought used in \citet{Du2021}, \citet{Zhou2017} and \citet{Huang2023}.
Hence, we use the underlying assumptions we have provided for our $(\alpha, \gamma)$-Channels to show $\epsilon$-DP for any channels that follow \autoref{def:alpha_gamma_channel}:
\begin{theorem}
    For all input states $\rho$ and $\sigma$, which are bounded by a trace distance $\tau(\rho, \sigma) \leq \tau_D$, any $(\alpha, \gamma)$-channel, $\mathcal{E}_{(\alpha, \gamma)}(\rho)$ suffices $\epsilon$-DP, with:
    $$
    \epsilon = \ln \left[1+\frac{\alpha \tau_D}{\gamma}\right]
    $$
\end{theorem}

\begin{proof}
We  assume the noisy probabilities as $\tilde{y}_k(\sigma) = \text{Tr}\left[\Pi_k \tilde{\mathcal{E}}(\sigma)\right]$ and the clean as $y_k(\sigma) = \text{Tr}\left[\Pi_k \mathcal{E}(\sigma)\right]$. 
We can note that the absolute difference between the two probabilities is bound by their sensitivity (See \citet[Section 8.1]{Angrisani2023}):
\begin{align*}
    | \tilde{y}_k(\sigma) - y_k(\rho) | & = \\
    | \text{Tr}\left[\Pi_k (\tilde{\mathcal{E}}(\sigma) - \tilde{\mathcal{E}}(\rho))\right] | & \leq \\
    \frac{1}{2} \lVert \Pi_k \rVert_\infty \lVert (\tilde{\mathcal{E}}(\sigma) - \tilde{\mathcal{E}}(\rho)) \rVert_1 & = \\
    \lVert \Pi_k \rVert_\infty \tau(\tilde{\mathcal{E}}(\sigma), \tilde{\mathcal{E}}(\rho)) & \overset{\text{POVM norm is one}}{=} \\
    \tau(\tilde{\mathcal{E}}(\sigma), \tilde{\mathcal{E}}(\rho)) & \overset{\text{\autoref{def:alpha_gamma_channel}}}{\leq} \\ 
    \alpha \tau(\sigma, \rho) 
\end{align*}
Secondly, we can note that for any POVM Measurement:
\begin{equation}
    y_k(\sigma) = \text{Tr} \left[\Pi_k \mathcal{E}_{(\alpha, \gamma)}(\sigma)\right] \geq \gamma
\end{equation}
With those two relations, we can derive an upper bound the DP relation:
\begin{align*}
    \frac{y_k(\rho)}{y_k(\sigma)} - 1 &= \frac{y_k(\rho) - y_k(\sigma)}{y_k(\sigma)} &\leq 
    \frac{\left| y_k(\rho) - y_k(\sigma) \right|}{y_k(\sigma)} \\ &\leq \frac{\alpha \tau(\sigma, \rho)}{\gamma} \leq\frac{\alpha \tau_D}{\gamma}  
\end{align*}

Given now:
\begin{equation}\label{eq:eps_and_delta}
    \epsilon = \ln \left[1+\frac{\alpha \tau_D}{\gamma}\right]
\end{equation}
We can resubstitute \autoref{eq:eps_and_delta} in the last relation above, which equates to the same term:
\begin{equation*}
    e^\epsilon - 1 = e^{ \ln \left[1+\frac{\alpha \tau_D}{\gamma}\right]} - 1  =  \frac{\alpha \tau_D}{\gamma} 
\end{equation*}
\end{proof}

\subsection{Relation of \texorpdfstring{$\alpha, \gamma$}{alpha, gamma}-channels to depolarizing channels}
Now that we have introduced an $(\alpha, \gamma)$-channel, it is necessary to consider whether the bound provided for $\epsilon$-DP is tight.
Specifically, given that the depolarizing channel is one of the most significant sources of noise, it is logical to demonstrate a bound that it incorporates the results of \citet{Du2021}.
In addition, to outline the extendibility of our framework, we exploit our concept of $(\alpha, \gamma)$-channels to include the same bound as \citet{Huang2023}.

\medskip

\subsubsection{Depolarizing Noise}\label{sec:depolarizing_noise}


One can find that the contraction of a depolarizing noise channel with value $p$ is defined as $\kappa = (1-p)$~\citet[Chapter 9]{nielsen_chuang}, from which we directly conclude that $\alpha = (1-p)$.
To upper bound the smallest eigenvalue of a noise channel, we follow the line of argument of \citet{Du2021}. 
Essentially they argue, that the smallest possible value of the noisy channel is given as $\text{Tr} \left[\Pi_k \mathcal{E}(\sigma)\right] = \nicefrac{p}{D}$.
Substituting this into our bound of $\epsilon$-DP, it is apparent that we arrive at the same bound as in \autoref{eq:dep_upperbound}:
\begin{equation}
    \epsilon = \ln \left[1 + \frac{(1-p) \tau_D}{\nicefrac{p}{D}}\right] = \ln \left[1 + D\frac{(1-p) \tau_D}{p}\right] 
\end{equation}

\subsubsection{Random Rotation Bound}
\citet{Huang2023} described their noise process as incorporating random rotations \textit{before} the operation of the quantum classifier. 
Given that this process still constitutes a quantum channel, we can safely set the trivial upper bound for $\alpha$ at $1$. 
To establish a bound on $\gamma$, we consider the scenario of a binary classifier~\footnote{A base assumption on the robustness in \citet{Huang2023}}. 
This assumption allows us to assume that the majority class must be at least greater or equal than $\frac{(1+h)^n}{2}$, with $h$ the lower bound on the random rotations in the $R_x$-Gates: 
\begin{equation}
    y_C(\sigma) > \frac{(1+h)^n}{2}.
\end{equation}
Hence, we assume w.l.o.g.:
\begin{equation*}
    y_C = \frac{(1+h)^n}{2} \geq 0.5.
\end{equation*}
From this, we deduce the following for the subsequent class:
\begin{equation*}
    y_{k \neq C} = 1 - \frac{(1+h)^n}{2} \leq \frac{(1+h)^n}{2},
\end{equation*}
and for all labels $y_C, y_{k \neq C} \geq \frac{(1+h)^n}{2}$.
\medskip

This provides us with $\gamma = \frac{(1+h)^n}{2}$ and by substituting $t = (h+1)\cdot2^{1-n}$, we can find $\gamma = t^n$.

In the end, setting $\alpha = 1$ and $\gamma = t^n$ in \autoref{eq:dep_upperbound} yield the same bound as \citet{Huang2023}:
\begin{equation*}
    \epsilon = \ln \left(\frac{\tau_D}{t^n} + 1\right)
\end{equation*}

\subsection{Certifiable robustness on \texorpdfstring{$(\alpha, \gamma)$}{alpha, gamma}-channels}

This section builds on the robustness bound from \citet{Du2021} and generalizes it to an $(\alpha, \gamma)$-channel. 
We focus on the infinite sampling case since we base our assumptions on a hypothetical model, which can be simulated but not implemented. 

\begin{theorem}[Infinite sampling case]\label{th:infinite_sampling}
    Given a binary classifier, altered by a $(\alpha, \gamma)$-channel: $y_k(\sigma) = \text{Tr} \left[ \Pi_k \mathcal{E}_{\alpha, \gamma}(\sigma) \right]$. Suppose that for a predicted class $C$, the following holds:
    \begin{equation*}
        y_C(\sigma) > e^{2\epsilon} y_{k \neq C}(\sigma) \qquad \epsilon = \ln \left[1+\frac{\alpha \tau_D}{\gamma}\right],
    \end{equation*}
    than for a benign state $\rho$ with $\tau(\sigma, \rho) \leq \tau_D$, the predicted class label $C$, will not change:
    \begin{equation*}
        \argmax_{k} y_k(\sigma) = C =  \argmax_{k} y_k(\rho)
    \end{equation*}
\end{theorem}

\begin{proof}
We can construct the proof similarly as in \citet{Du2021}, by simply replacing the definition of $\epsilon$.
We will first note that due to $\epsilon$-DP:
\begin{equation}
    y_C(\rho) \geq e^{-\epsilon} y_C(\sigma)
\end{equation}
Hence using $ y_C(\sigma) > e^{2\epsilon} y_{k \neq C}(\sigma)$:
\begin{equation}
     y_C(\rho) \geq e^{-\epsilon} y_C(\sigma) > e^{\epsilon} y_{k \neq C}(\sigma)
\end{equation}
Next, we can note that:
\begin{equation}
    y_{k \neq C}(\rho) \leq e^\epsilon y_{k \neq C}(\sigma)
\end{equation}
Which completes the relation:
\begin{equation}
     y_C(\rho) \geq e^{-\epsilon} y_C(\sigma) > e^{\epsilon} y_{k \neq C}(\sigma) \leq  y_{k \neq C}(\rho)
\end{equation}
Since $y_C(\rho) > y_{k \neq C}(\rho)$, we can conclude that:
\begin{equation}
    \argmax_{k} y_k(\sigma) = C =  \argmax_{k} y_k(\rho)
\end{equation}
\end{proof}

Based on Theorem~\ref{th:infinite_sampling}, if the ratio $\frac{\tilde{y}_C(\sigma)}{\tilde{y}_{C \neq k}(\sigma)}$ exceeds $e^{2\epsilon}$, we can certify the data point against any benign sample within radius $\tau_D$. 
This leads to the following corollary:
\begin{corollary}\label{corr:cert_distance}
    Given a \textit{binary} quantum classifier perturbed by a $(\alpha, \gamma)$-channel, with:
    \begin{equation*}
        B \equiv \frac{\tilde{y}_C(\sigma)}{\tilde{y}_{C \neq k}(\sigma)}
    \end{equation*}
    This classifier is robust against any perturbation: $\sigma \rightarrow \rho$, with $\tau(\sigma, \rho) \leq \tau_D$, if:
    \begin{equation*}
        \frac{\tilde{y}_C(\sigma)}{\tilde{y}_{C \neq k}(\sigma)} = B > e^{2\epsilon}
    \end{equation*}
    With: $\epsilon = \ln \left[1+\frac{\alpha \tau_D}{\gamma}\right]$
\end{corollary}

Building on our prior analysis of depolarizing noise (\ref{sec:depolarizing_noise}), where we established a certifiable distance for a single data point based on network output probabilities, we can now determine the minimum value of $\epsilon$ that bounds the fraction $B$. This involves calculating:
\begin{equation}\label{eq:eminlowerbound}
    \epsilon_{min} = \frac{1}{2}\ln \left( \nicefrac{y_C(\sigma)}{y_{C \neq k}(\sigma)} \right)
\end{equation}
Next, one can simply use $\epsilon = \ln \left[1+\frac{\alpha \tau_D}{\gamma}\right]$, which will lead to a certifiable distance defined as:
\begin{equation}\label{eq:cert_distance}
    \tau_D = (e^{\epsilon_{min}}-1) \cdot \frac{\gamma}{\alpha}
\end{equation}
Note that for $\alpha = 0$, we would get a $\tau_D = \infty$, which highlights the limitations of our approach and the necessity of incorporating $(\epsilon, \delta)$-DP into the analysis of $(\alpha, \gamma)$-channels.

\section{Construction of optimal (\texorpdfstring{$\alpha$}{alpha}, \texorpdfstring{$\gamma$}{gamma})-Channels}

In this section, we investigate optimal noise channels for adversarial robustness. 
We formulate a semidefinite program (SDP) to find the best channel based on the certification distance from~\autoref{corr:cert_distance}.
\medskip

We consider a single data point (later expanded to a set) represented by a density matrix $\sigma = \sigma(x)$. 
A pretrained quantum classifier described by a unitary $U$ acting on $n$-qubits and measured on $d$ qubits by $2^d$ POVM-measurements: $\{\Pi_k\}$. 
We assume a true label $C$ (w.l.o.g.).
For brevity, we denote $\sigma_D = U \sigma U^\dagger$ and abbreviate this by $\sigma$.
Each class propability is defined by the POVM-Measurement: 
$y_k(\sigma) = \text{Tr} \left[ \Pi_k \mathcal{E}_{(\alpha, \gamma)}(\sigma) \right]$,
where $\mathcal{E}_{(\alpha, \gamma)}$ represents the noise channel with parameters $\alpha$ and $\gamma$.
Given a threshold $B = \frac{y_C(\sigma)}{y_{k \neq C}(\sigma)} > e^{2\epsilon}$, we recognize that $\epsilon$ is fixed by the chosen noise channel. 
Hence, to maximize classifier robustness, a natural approach is to maximize the fraction $B$.
More specifically we can formulate this as the following optimization problem:
\begin{equation}
    \begin{aligned}
        \max \quad & \quad \nicefrac{y_C(\sigma)}{y_{k \neq C}(\sigma)}, \\
        \text{s.t:} \quad 
        & y_k(\sigma) = \text{Tr} \left[\Pi_k \mathcal{E}(\sigma)\right], \\
        & \mathcal{E}( . ) \; \text{is a quantum channel}, \\
        & \forall \sigma, \rho \quad \tau(\mathcal{E}(\sigma), \mathcal{E}(\rho)) \leq \alpha \tau(\sigma, \rho) + \beta, \\
        & \forall k \qquad y_k(\sigma) \geq \gamma.
    \end{aligned}
\end{equation}

Unfortunately, the current problem formulation violates the rules of convexity and might not be directly solvable using SDPs. 
To address this, in the following subsections, we introduce reformulations to make it suitable for SDPs.

\medskip

\subsubsection{CPTP Properties}\label{sec:channel}

An underlying assumption of this work is that quantum channels can be described by Choi matrices, which must be CPTP. 
As pointed out earlier, we use the Choi–Jamiołkowski isomorphism~\cite{CHOI1975285}, since kraus operators do not necessarly provide a convex or linear representation of our variables, which is a necessity in SDPs. 
We follow \citet{Knee_2018} and encode the two properties (CP and TP) individually as follows:
\begin{align}
    \je  \succeq 0 \label{eq:cp} \\
    \text{Tr}_\text{out} \left[ \je  \right] = I_A \label{eq:tp}
\end{align}
\medskip

\subsubsection{\texorpdfstring{$\alpha$}{alpha}-Bound for the SDP}\label{sec:const1}
To constrain the $\alpha$ parameter, we leverage \citet[Theorem 8.17]{wolf2012quantum}. 
While the original source includes a disclaimer about potential inaccuracies, we provide a complete proof in Appendix~\ref{apx:proof_theorem817} for clarity.
Here, we directly utilize \citet[Theorem 8.17]{wolf2012quantum} to bound the contraction $\alpha = (1-\kappa)$ within our SDP formulation, as shown below:
\begin{equation}
    \mathcal{E} - \kappa \mathcal{E}^\prime \succeq 0.
\end{equation}
Here $\mathcal{E}$ is our noise channel and $\mathcal{E}^\prime$ describes a reference channel. \citet{wolf2012quantum} suggests to set: $\mathcal{E}^\prime = \nicefrac{I}{d}$ the fully mixed state as a reference, which will yield the following constraint on the Choi matrix of our noise channel:
\begin{equation}
    \je - \frac{\kappa}{d} I \succeq 0
\end{equation}
While the current constraint provides a valid upper bound, it has limitations. 
Given that the smallest eigenvalue of $\je$ is zero, this bound offers no advantage over the trivial case $\kappa = 0 \implies (1-\kappa) = 1$. 
Nonetheless, due to the lack of better methods for encoding the contraction factor within SDPs, we employ this approach, acknowledging that it restricts the representable noise channels.

\subsection{\texorpdfstring{$\gamma$}{gamma}-Bound for the SDP}\label{sec:const2}
To enforce the constraint on $\gamma$, we refrained from solving the SDP for all possible density matrices due to the clear in applicability of encoding a continuous set in a discrete setting. Henceforth, we achieve a similar behavior by encoding the following constraint in our PSD:
\begin{equation}
    \forall \Pi_k \quad \text{Tr}_2 \left[(\Pi_k \otimes I) \je \right] - \gamma I \succeq 0
\end{equation}
Since any POVM measurement in the context of Choi matrices is defined as in \autoref{eq:povm_choi}, we can follow from the constraint, that: $\text{Tr}_2 \left[ (I \otimes \rho^T)\je \right] \geq \gamma I $, which we can use to verify that all POVM measurements are bounded by $\gamma$: 
\begin{align*}
    y_k(\sigma)   &= \text{Tr} \left[(\Pi_k \otimes \rho^T) \je\right] \\
                &= \text{Tr} \left[ \Pi_k \text{Tr}_2 \left[ (I \otimes \sigma^T)\je \right]\right] \\
                & \geq \text{Tr} \left[ \Pi_k \gamma I\right] \\
                & = \gamma 
\end{align*}

\subsection{Combination of the SDP formulation}

Given the constraints above, we can formulate an SDP to optimize for the best possible noise channel. 
Note that the function $\nicefrac{y_C(\rho)}{y_{C \neq k}(\rho)}$ is non-convex and therefore not applicable in the context of SDP formulations. 
Hence, we simply maximize $y_C(\rho)$, since for the binary case $y_{C \neq k}(\rho) = 1-y_C(\rho)$, hence the target optimization goal should still be met.
With the constraints established, we can now formulate an SDP to find the optimal noise channel. However, the original objective function involving the ratio $\frac{y_C(\rho)}{y_{C \neq k}(\rho)}$ is non-convex and unsuitable for SDPs.
To address this, we propose a simpler approach: maximize $y_C(\rho)$. 
Since in binary classification, the probabilities for the two classes sum to 1 ($y_{C \neq k}(\rho) = 1 - y_C(\rho)$), 
than, maximizing the probability $y_C(\rho)$ for the correct class inherently minimizes the probability for the incorrect class.
\medskip

So far we have only assumed one data sample in our line of argument, since solving the SDP for every possible input will correspond to label leaking, we decided to construct the noise channel over the sum of all possible probabilities $y_{C_i}(\sigma_i)$ in the test set. We have chosen the test set to ensure that the noise channel is optimal for the samples attacked in our evaluation and no assumption about the generality have to be made. Next, we consider the class imbalance present on the test set by computing the weight vector $w = \frac{\sum_k C_k}{\lvert C \rvert}$, which leads to the overall optimization target:$\sum_i w_i \cdot y_{C_i}(\sigma_i)$.

During our evaluation, we found overall two scenarios that can be solved using our SDP: At first, one can apply the noise channel \textit{after} the quantum classifier (post-order SDP) and \textit{infront of} the quantum classifier (pre-order SDP). Both scenarious are outlined in the following.
\medskip

\subsubsection{Post order SDP}


The post order SDP is rather a straight forward application of the previously discussed concepts. 
\begin{equation}\label{eq:post_order_sdp}
    \begin{aligned}
        \max \quad & \quad \sum_i w_i \cdot y_{C_i}(\sigma_i) \\
        \text{s.t:} \quad 
        & \forall i \quad y_{C_i}(\sigma_i) = \text{Tr} \left[(U \sigma_k U^\dagger \otimes \Pi_{C_k}) \je  \right] \\
        & \je  \succeq 0 \\
        & \text{Tr}_2 \left[ \je  \right] = I_A \\
        & \forall \Pi_k \quad \text{Tr}_2 \left[(\Pi_k \otimes I) \je \right] - \gamma I \succeq 0 \\
        & \je  - \nicefrac{\alpha}{D} I  \succeq 0 
    \end{aligned}
\end{equation}
\medskip

\subsubsection{Pre-order SDP}
With the pre-order SDP, we can note that the unitary will "act" on the Choi matrix upon measurements, which will alter our measurements in the following way:
\begin{equation}
    y_k(\sigma) = \text{Tr} \left[ (\sigma \otimes \Pi_{k}) (I \otimes U) \je (I \otimes U^\dagger) \right]
\end{equation}
Interestingly, this also implies for our $\gamma$ constraint that:
\begin{equation}
    \text{Tr}_2 \left[(\Pi_k \otimes I)(I \otimes U) \je (I \otimes U^\dagger) \right] - \gamma I \succeq 0
\end{equation}
Aside from this change, the SDP formulation will be unchanged leading us to the formulation:
\begin{equation}\label{eq:pre_order_sdp}
    \begin{aligned}
        \max \quad & \quad \sum_i w_i \cdot y_{C_i}(\sigma_i) \\
        \text{s.t:} \quad 
        & \forall i \quad y_{C_i}(\sigma_i) = \text{Tr} \left[ [(\sigma_i \otimes \Pi_{C_i}) (I \otimes U) \je (I \otimes U^\dagger) \right] \\
        & \je  \succeq 0 \\
        & \text{Tr}_2 \left[ \je  \right] = I_A \\
        & \forall \Pi_k \quad \text{Tr}_2 \left[(\Pi_k \otimes I)(I \otimes U) \je (I \otimes U^\dagger) \right] - \gamma I \succeq 0 \\
        & \je  - \nicefrac{\alpha}{D} I  \succeq 0 
    \end{aligned}
\end{equation}


\section{Experiments}
In our experiments, we wanted to highlight the utility and limitations of utilizing a SDP Formulation as a bechmark for optimally robust quantum channels. See Appendix~\ref{apx:experimental_setup} for a detailed description of our data preprocessing and SDP implementation steps. We have considered the Iris dataset (\textbf{Iris}) with \textit{Angle} and \textit{Amplitude} embedding, the Pima Indians Diabetes dataset (\textbf{PID}) and the breast cancer dataset (\textbf{BC}). The latter two where only utilized using Amplitude embedding.

\smallskip

In our experiments, we want to highlight the performance of depolarizing noise against our optimal $(\alpha, \gamma)$ counter part. Further we want to analyze the effect of $\alpha$ and $\gamma$ on the certifiablility and at last we want to compare the effectiveness of different encodings.

\subsection{Classification of Depolarizing Noise against optimal noise}
Providing an optimal noise channel for various values of depolarizing noise serves as an experimental benchmark for the optimality of specific noise channels. 
Here, we outline the process. 
As a first step, one needs to reduce $\alpha$ and $\gamma$ to the parameters of the specific noise channel. 
Next, one can iterate over a set of parameters and compare the adversarial accuracy to that of the optimal noise channel, as illustrated in \autoref{fig:adv_acc}. 
For depolarizing noise, we observe that in most datasets, a small amount of depolarizing noise provides a good upper bound, but typically the effect decays with an increase in depolarizing noise.
\begin{figure*}[h!tb]
        \centering
        \begin{subfigure}{0.32\linewidth}
            \centering
            \includegraphics[width=\textwidth]{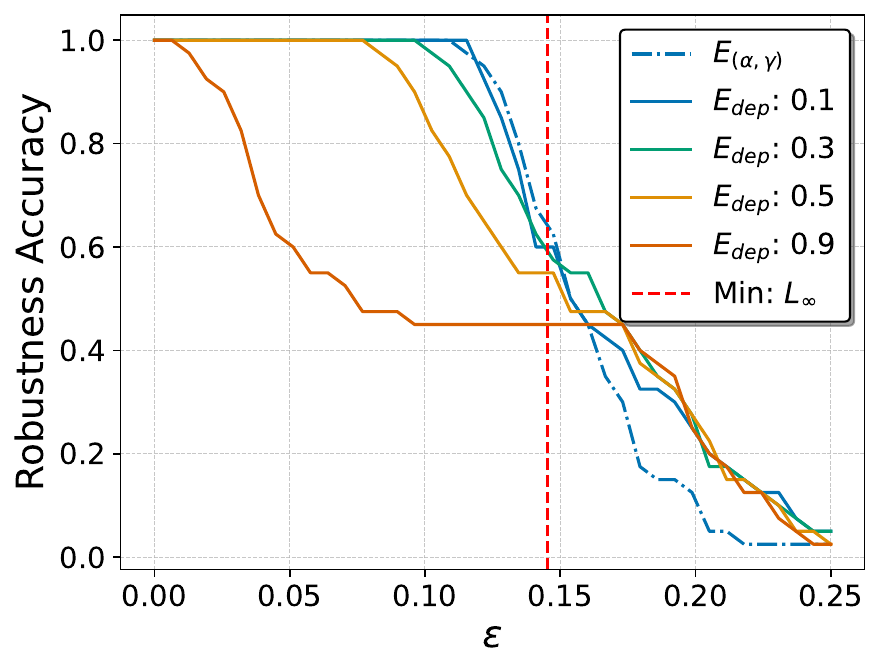}
            \caption{Iris (Amplitude)}
        \end{subfigure}
        \begin{subfigure}{0.32\linewidth}
            \centering
            \includegraphics[width=\textwidth]{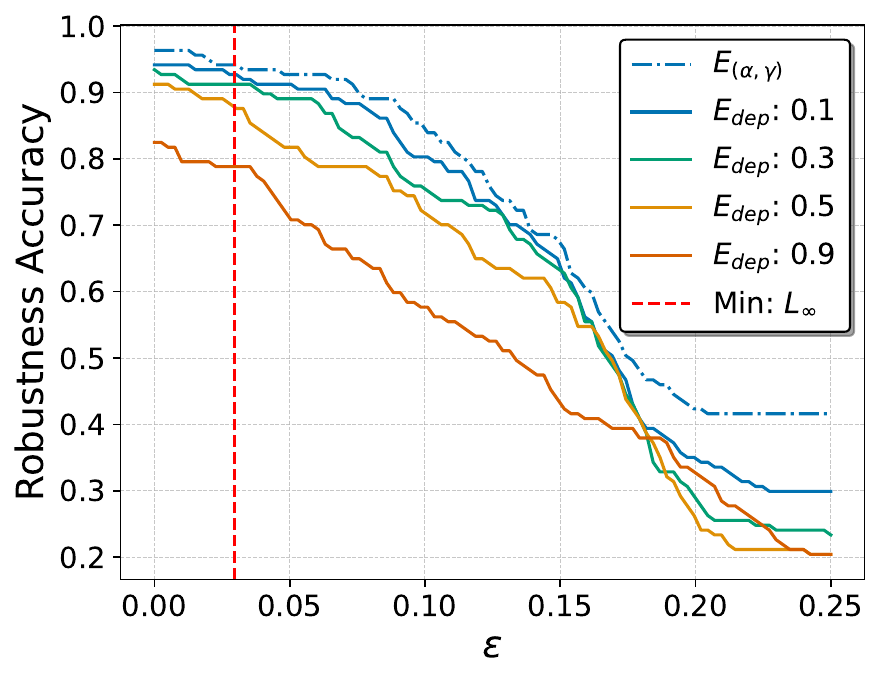}
            \caption{BC (Amplitude)}
        \end{subfigure}
        \begin{subfigure}{0.32\linewidth}
            \centering
            \includegraphics[width=\textwidth]{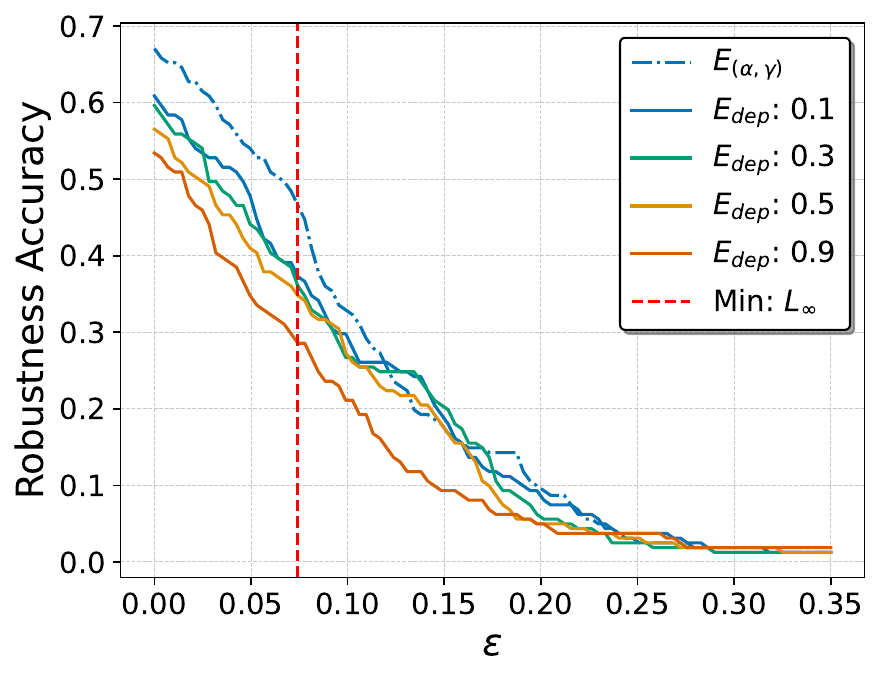}
            \caption{PID (Amplitude)}
        \end{subfigure}
    \caption{
    The plot shows adversarial accuracy under an FGSM attack across all evaluated datasets: the blue dashed curve, \(E_{(\alpha,\gamma)}\), represents adversarial accuracy with our best-case noise channel, the red dashed vertical line indicates the smallest \(L_\infty\) distance between any two samples from opposite classes, and the colored curves \(E_{dep}\) depict adversarial accuracy when a depolarization noise channel is applied at increasing noise levels. We include only one \((\alpha,\gamma)\)-channel because the individual channels largely overlap (See \autoref{apx:coincideness_of_optimal_channels}).}
    \label{fig:adv_acc}
\end{figure*}

\subsection{Influence on \texorpdfstring{$\alpha$}{alpha} and \texorpdfstring{$\gamma$}{gamma} on the certifiable robustness}
Next, we want to outline the effect of $\alpha$ and $\gamma$ on the certifiable accuracy. Therefore, we consider \autoref{eq:cert_distance} to compute the certifiable distance for each datasample in terms of the trace distance and then compute the portion of certifiable samples on the test set given a set of distances $\tau_D \in \{0.05, 0.10, 0.15\}$.
We first note that the effect of $\alpha$ and $\gamma$ on the difference between the class label and the follow up class (which is later on described in the fraction B in \autoref{corr:cert_distance}) scale symetrically with $\alpha$ and $\gamma$. 
See \autoref{fig:ydiff} for a visualization. 
Additionally one can observe that the difference always converge to zero for all datasets given: $\alpha = 0$ and $\gamma = 0.5$. 
For these values the quantum classifer is perturbed by a noise channel describing random guessing, which yields naturally a certifiable distance of $\tau_D = \infty$. 
Given the limited usefulness of a random guessing network, we decide to constrain the certifiable distance shown by the plots in \autoref{fig:cert_distances} by a $\delta$ on the difference among $y$ labels. 
To this end, we choose a $\delta$ of $0.05$ for the Iris dataset (Angle and Amplitude Embedding) and the BC dataset. For the PID dataset, we pick a $\delta$ of $0.01$, since the class probabilities are closed together due to the lower accuracy. 

\begin{figure}[h!]
    \centering
    \begin{subfigure}{0.45\linewidth}
        \centering
        \includegraphics[width=\textwidth]{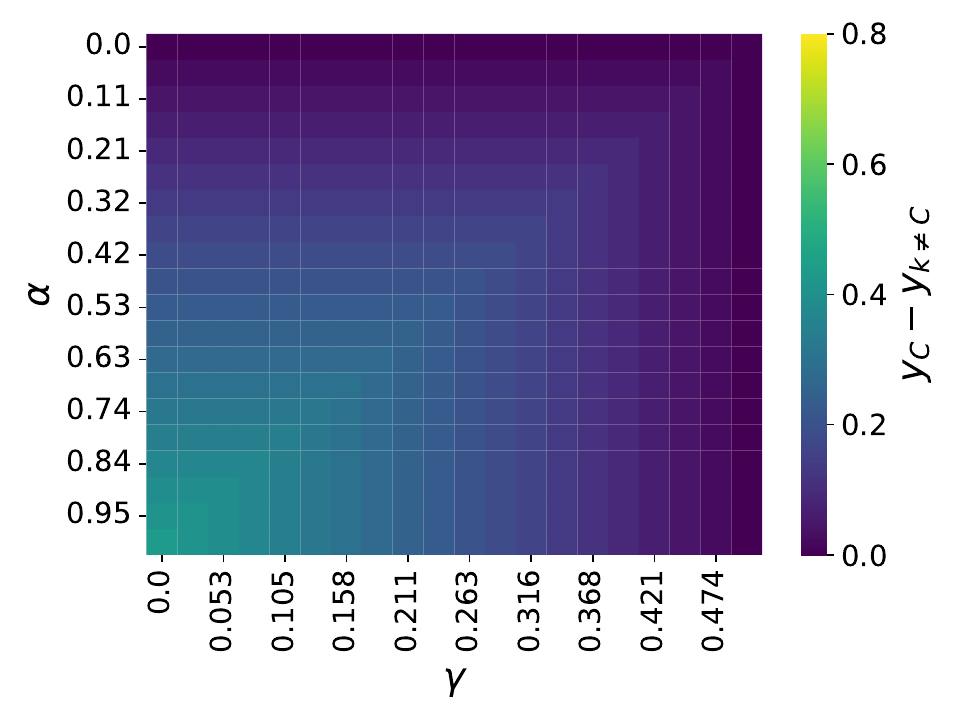}
        \caption{Iris (Amplitude)}
    \end{subfigure}
    \begin{subfigure}{0.45\linewidth}
        \centering
        \includegraphics[width=\textwidth]{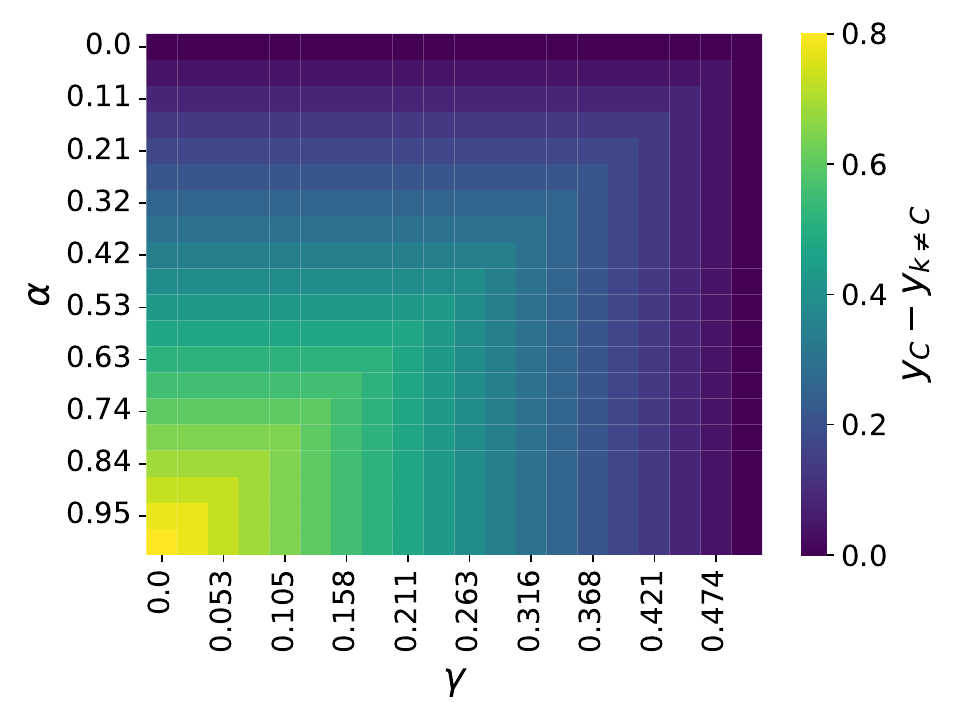}
        \caption{Iris (Angle)}
    \end{subfigure}
    \begin{subfigure}{0.45\linewidth}
        \centering
        \includegraphics[width=\textwidth]{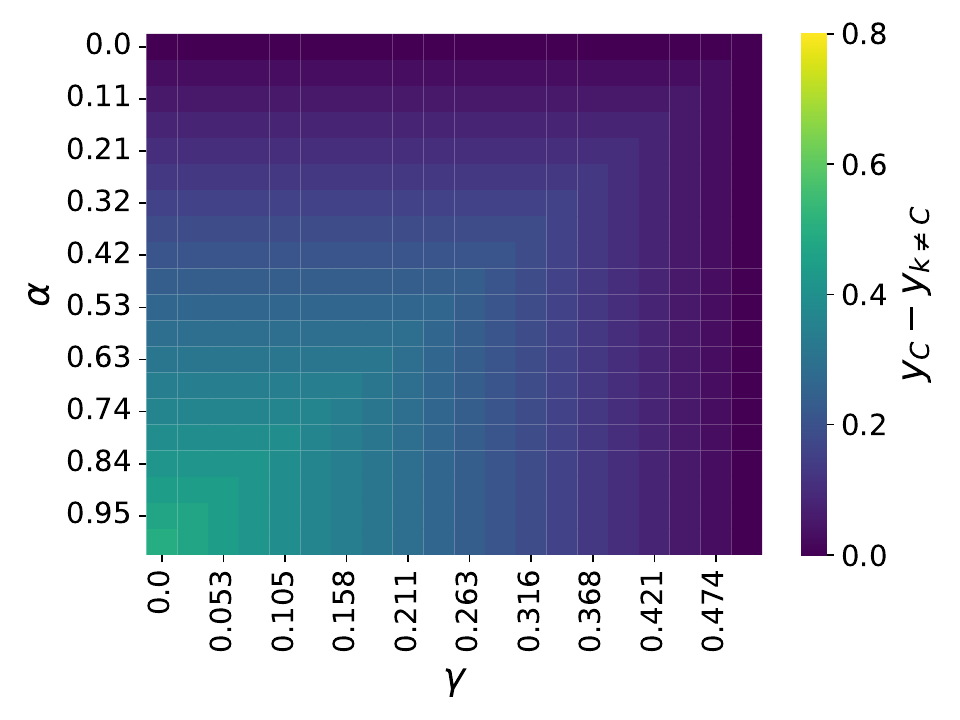}
        \caption{BC (Amplitude)}
    \end{subfigure}
    \begin{subfigure}{0.45\linewidth}
        \centering
        \includegraphics[width=\textwidth]{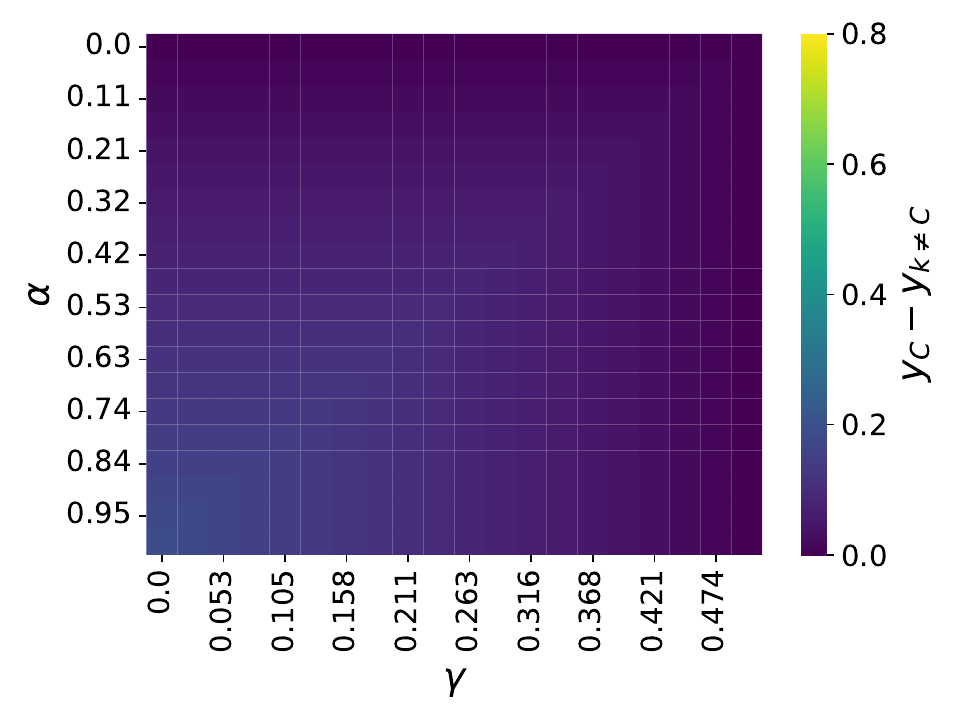}
        \caption{PID (Amplitude)}
    \end{subfigure}
    \caption{Difference $y_C - y_{k \neq C}$ for the various different datasets used in this study. }
    \label{fig:ydiff}
\end{figure}
\medskip

Concerning the certifiable samples in \autoref{fig:cert_distances}, one can see that the influence of $\alpha$ and $\gamma$ have an effect on the amount of certifiable distances given by the classifier. 
The area where the maximum of samples can be certified clearly shrinks in case the cut-off distance is increased from $0.05$ to $0.15$ and typically converges to a line like structure which is the maximum for $(\alpha = 0.11, \gamma = 0.447)$ (PID, Iris Angle and BC), $(\alpha = 0.16, \gamma = 0.421)$ (Iris Amplitude). Which additionally describes the smallest allowed difference considered in our experiments. 

Interestingly, across various datasets, lower values of $\gamma$ tend to extend the range of certifiable samples. These samples maintain their high certification level for a longer period before experiencing a characteristic sharp drop. This drop point depends on the chosen certifiable distance, the specific dataset, and the embedding used. In contrast, the parameter $\alpha$ seems to induce a smoother transition in the certification level. It remains an open question whether this observed behavior is a consequence of the SDP formulation itself or an inherent property of $(\alpha, \gamma)$-channels.

\subsection{Comparison of amplitude and angle embedding}
In \autoref{def:alpha_gamma_channel}, it is apparent, that the $\epsilon$ of the $\epsilon$-DP mechanism is proportional to $\alpha$, $\beta$ and the trace distance of the two input states, which highlights that different encodings, might additionally provide further robustness bounds on quantum classifiers. To provide further insights, we have normalized the $\epsilon$ values used in this study from the respective intervals of $[0;0.25]$ (Amplitude embedding) and $[0;0.9]$ (Angle Embedding) of the Iris dataset to the interval $[0;1]$. The results can be observed in \autoref{fig:amplitude_angle_comparison}. Overall it is apparent, that a better encoding shows a higher degree of adversarial accuracy, which decays at a similar rate with an increase in depolarizing noise and is still bounded by the optimally constructed $(\alpha, \gamma)$-Channel. 
\begin{figure}[h!]
    \centering
    \includegraphics[width=\linewidth]{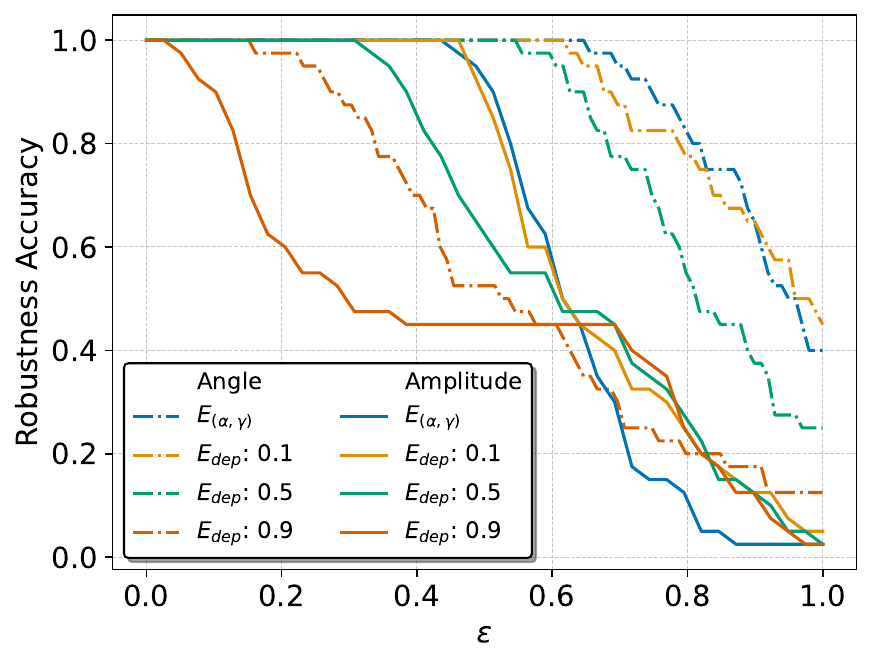}
    \caption{Normalized adversarial accuracy's between amplitude and angle encoding. The adversarial attacks where done using FGSM with a $L_\infty$-norm. The blue dashed curve, \(E_{(\alpha,\gamma)}\), represents adversarial accuracy with our best-case noise channel, and the colored curves \(E_{dep}\) indicate adversarial accuracy when a depolarization noise channel is applied at increasing noise levels.}
    \label{fig:amplitude_angle_comparison}
\end{figure}

\begin{figure*}
    \centering
    \begin{subfigure}{\linewidth}
        \includegraphics[width=\textwidth]{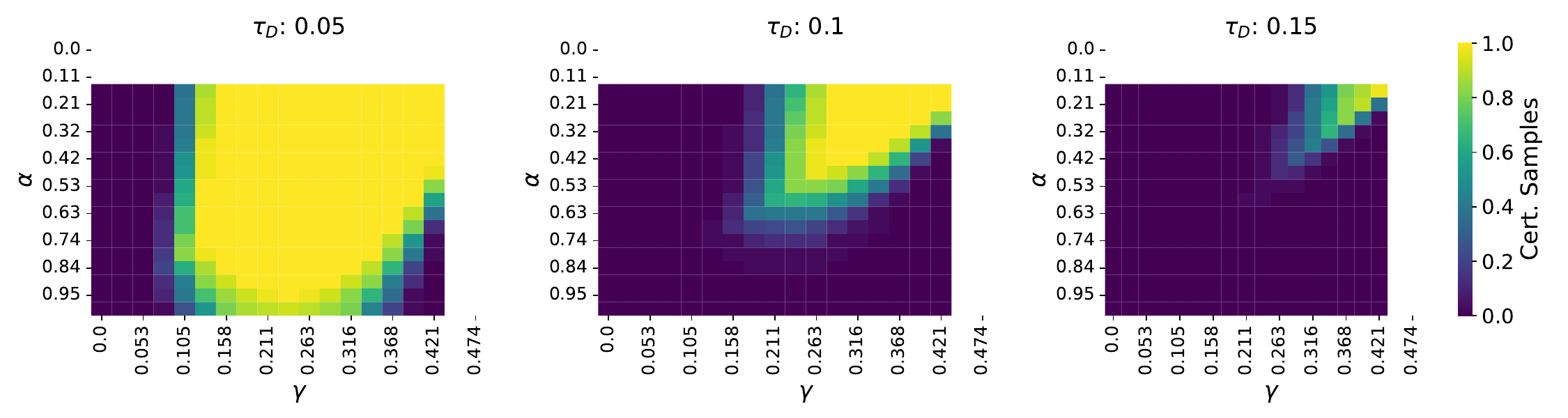}
        \caption{Iris (Amplitude)}
    \end{subfigure}
    \begin{subfigure}{\linewidth}
        \includegraphics[width=\textwidth]{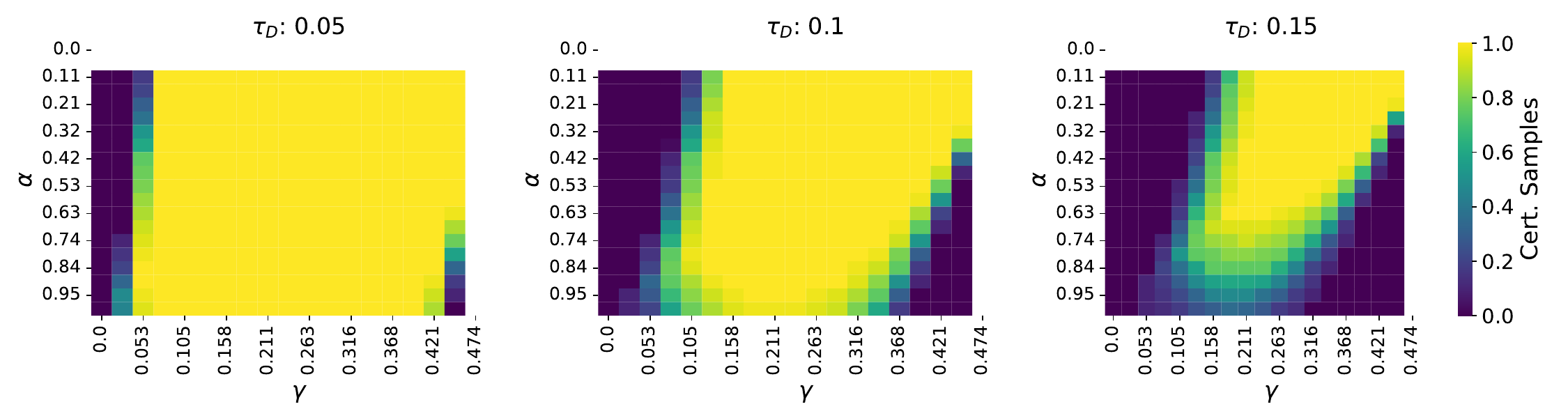}
        \caption{Iris (Angle)}
    \end{subfigure}
    \begin{subfigure}{\linewidth}
        \includegraphics[width=\textwidth]{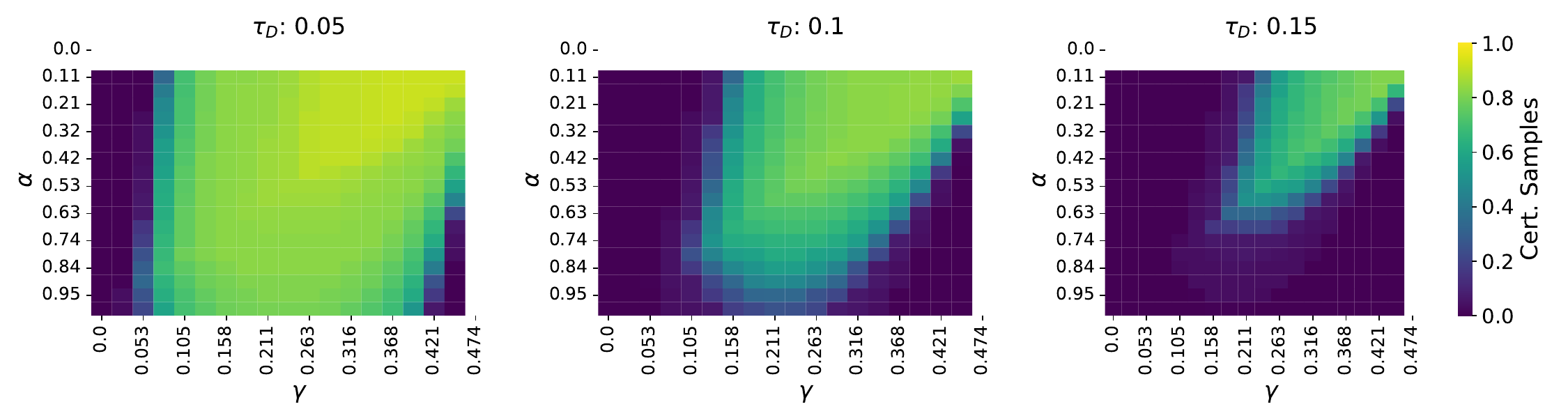}
        \caption{BC (Amplitude)}
    \end{subfigure}
    \begin{subfigure}{\linewidth}
        \includegraphics[width=\textwidth]{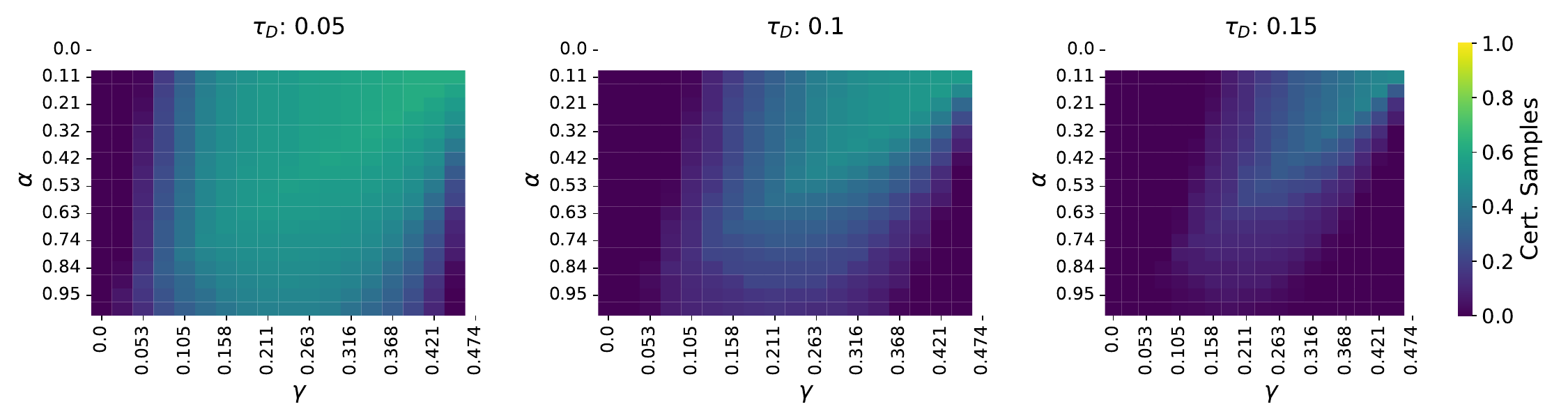}
        \caption{PID (Amplitude)}
    \end{subfigure}
    \caption{Portion of certifieable distances among the test set for different values of $\alpha$ and $\gamma$. The portion is computed by grouping among the levels: $\{0.05, 0.10, 0.15\}$. }
    \label{fig:cert_distances}
\end{figure*}

\section{Conclusion}
In this work, we have defined a family of noise channels, which are certifiable robust. We have hereby related robustness of QNNs to the contraction of the algorithm, to the smallest possible label assignment as well as to the encoding strategy used by the quantum classifier. 
Our SDP formulation, while theoretically useful for constructing optimal noise channels, suffers from numerical scaling limitations in practical applications.
\medskip

Our framework could be extended to evaluate other noise sources beyond depolarizing noise, potentially leading to a broader understanding of quantum noise and its impact on $\epsilon$-DP.
Next, while a small amount of depolarizing noise shows promise as a defense against evasion attacks, it remains unclear if noise inherent to NISQ devices can truly provide robustness. 
Our findings, along with previous work, suggest that increasing depolarization leads to decreased certifiability.
\medskip

The concept of noise channels as a general model for quantum computation, coupled with the existence of certifiable robust channels, opens exciting avenues. Constructing a variational ansatz with inherently certifiable behavior would be a significant advance. Such an ansatz could be easier to optimize compared to finding the \textit{optimal} noise model (a convex optimization problem), despite the computational challenges arising from exponentially growing Hilbert spaces. This could pave the way for theoretically grounded adversarial robustness in quantum classifiers. 
\medskip

At last we want to emphasize that based on this study optimizing for a robust noise channel yields a minor advantage compared to the larger increase of the robustness budget when scaling the amount of qubits.

\appendices

\section{Proof of \citet[Theorem 8.17]{wolf2012quantum}}\label{apx:proof_theorem817}
In the following \citet[Theorem 8.17]{wolf2012quantum} is outlined combined with its proof and the accompanying references.
\begin{theorem}\label{th:tighter_contraction_bound}
    Given $\mathcal{E}, \mathcal{E^\prime}$ be two trace preserving and Hermitian preserving linear maps, where we define $\mathcal{E}^\prime$ as:
    \begin{equation*}
        \mathcal{E}^\prime(\rho) = \text{Tr} \left[\rho\right] Y; \quad \text{Tr} \left[Y\right] = 1
    \end{equation*}
    If $\mathcal{E} - \kappa \mathcal{E}^\prime$ is positive for some $\kappa \geq 0$, then for all density matrices $\rho, \sigma$:
    \begin{equation*}
        \text{Tr} \lvert \mathcal{E}(\sigma-\rho)\rvert \leq (1-\kappa) \text{Tr} \lvert \sigma-\rho \rvert
    \end{equation*}
\end{theorem}

\begin{proof}
We start off by noting that~\footnote{Since $\text{Tr}\left[\sigma - \rho\right] = \text{Tr}\left[\sigma\right] - \text{Tr}\left[\rho \right] = 1-1 = 0$} $\mathcal{E}^\prime(\sigma-\rho) = \text{Tr}\left[\sigma - \rho\right] Y = 0$. We can hence reformulate the trace distance as follows:
\begin{align*}
    \frac{1}{2} \text{Tr} \lvert \mathcal{E}(\sigma-\rho)\rvert & = \\
    \sup_{0 \leq P \leq I} \left\{\text{Tr} \left[ \mathcal{E}(\sigma-\rho) P\right]\right\} & = \\
    \sup_{0 \leq P \leq I} \left\{\text{Tr} \left[ (\mathcal{E}(\sigma-\rho) - \kappa \mathcal{E}^\prime(\sigma-\rho)) P\right]\right\} & = \\
    \sup_{0 \leq P \leq I} \left\{\text{Tr} \left[ ((\mathcal{E}- \kappa \mathcal{E}^\prime)(\sigma-\rho)) P\right]\right\} & 
\end{align*}
Next, we will define the dual map, such that:
\begin{equation}
    \text{Tr} \left[ ((\mathcal{E}- \kappa \mathcal{E}^\prime)(\sigma-\rho)) P\right] = \text{Tr} \left[ ((\mathcal{E}- \kappa \mathcal{E}^\prime)^*(P)) (\sigma-\rho)\right]
\end{equation}

Since $\mathcal{E}- \kappa \mathcal{E}^\prime$ is defined as positive, we can conclude that $(\mathcal{E}- \kappa \mathcal{E}^\prime)^*$ is positive too~\cite[Proposition 2.18]{Watrous_2018}. we can follow that since $P$ is positive and bounded by $I$, that $P^\prime = (\mathcal{E}- \kappa \mathcal{E}^\prime)^*(P)$ is positive. Since $\mathcal{E}$ and $\mathcal{E}^\prime$ are trace preserving, the application of identity to their dual channel scales $I$ as: $(\mathcal{E}- \kappa \mathcal{E}^\prime)^*(I) = (1-\kappa) I$. Given now that by definition $P \leq I$, we can follow that
\begin{equation}
    (\mathcal{E}- \kappa \mathcal{E}^\prime)^*(I) = (1-\kappa) I \geq P^\prime = (\mathcal{E}- \kappa \mathcal{E}^\prime)^*(P)
\end{equation}
We can follow, that the resulting projection is bounded by: $0 \leq P^\prime \leq (1-\kappa) I$, which allows the following relation:
\begin{align*}
    \sup_{0 \leq P \leq I} \left\{\text{Tr} \left[ ((\mathcal{E}- \kappa \mathcal{E}^\prime)(\sigma-\rho)) P\right]\right\} & = \\
    \sup_{0 \leq P \leq I} \left\{\text{Tr} \left[ ((\mathcal{E}- \kappa \mathcal{E}^\prime)^*(P)) (\sigma-\rho)\right]\right\} & \leq \\
    \sup_{0 \leq P^\prime \leq (1-\kappa)I} \left\{\text{Tr} \left[(\sigma-\rho) P^\prime\right]\right\} & = \\
    (1-\kappa) \sup_{0 \leq P^\prime \leq I} \left\{\text{Tr} \left[(\sigma-\rho) P^\prime\right]\right\} & = \\
    \frac{1-\kappa}{2} \text{Tr} \lvert \sigma-\rho\rvert
\end{align*}
From which we can follow, that:
\begin{equation}
    \frac{1}{2} \text{Tr} \lvert \mathcal{E}(\sigma-\rho)\rvert \leq \frac{1-\kappa}{2} \text{Tr} \lvert \sigma-\rho\rvert
\end{equation}
\end{proof}

\section{Experimental Setup}\label{apx:experimental_setup}

\subsection{Training of Neural networks}
We trained four different QNNs, using a similar scheme as in our previous study. Overall we used the Pima Indians Diabetes Dataset, the Breast Cancer and the Iris dataset. We used CrossEntropy as a loss function for all datasets and considered the class imblanance by adding weights to the loss. As a optimizer we used Adam.

\subsection{Data preprocessing} 
As already pointed out we used overall three datasets \texttt{Iris}, \texttt{Pima Indians Diabetes} and \texttt{Breast Cancer}.
\medskip

\subsubsection{Iris} For Iris we followed the approach by \citet{Du2021}, dropping the \texttt{petal\_width} feature as well as the \texttt{versicolor} class. We then normalized each datapoint by its respective $\lvert \lvert . \lvert \lvert_2$-Norm and performed a train/test split of $40:60$. Regarding the Angle Embedding, the datapoints where min max normalized to the interval $[0;\pi]$.
\medskip

\subsubsection{Pima Indians Diabetes (PID)}
For the Pima Indians Diabetes dataset, we first dropped all nan or duplicate features from the dataset, next to overcome the present class imbalance, we removed overall $232$ datapoints from the majority class at random. We used the remaining dataset with $536$ datapoints and split them into a train set consisting of $375$ datapoints and a test set of $161$ datapoints. To provide a comparabilty with the Iris dataset, we additionally normalized each datapoint with it's respective $\lvert \lvert . \lvert \lvert_2$-Norm.
\medskip

\subsubsection{Breast Cancer Wisconsin (BC)}
For the Breast Cancer Wisconsin dataset we found overall $683$ features, with a dimension of $9$, since this would require four qubits in the angle embedding, we performed a PCA with eight components, effectively rescaling the dimensionality to eight. 
We rescaled the data into the interval $\left(0;1\right]$ and normalized each feature vector by it's respective $\lvert \lvert . \lvert \lvert_2$-Norm. No further preprocessing was conducted.

\subsection{QNN Architecture}
As pointed out earlier, we used a similar scheme of dataembedding plus strongly entangling layers. For embedding we trained Iris once using angle embedding and once amplitude embedding. Since we where only able to solve up to 3 qubits, we refrained from using qngle embedding for the PID and BC datasets and applied hereby amplitude embedding. For an overview of the various hyperparameters used in this study, refer to \autoref{tab:hyperparams}.

\begin{table}[hbt!]
 \caption{Hyperparameters used for the various types of QNNs for the datasets Iris, Pima Indians Diabetes (PID) and Breast Cancer (BC).}\label{tab:hyperparams}
\begin{adjustbox}{width=\linewidth,center}
    \begin{tabular}{lllllll}
        \toprule
        Datset & Embedding & Qubits     & Layers     & Batch      & Learn. Rate   & Epochs \\
        \midrule
        Iris   & Amplitude & 2          & 2          & 30         & 0.05          & 100        \\
        Iris   & Angle     & 3          & 2          & 30         & 0.01          & 100        \\
        BC     & Amplitude & 3          & 40         & 16         & 0.0005        & 10         \\
        PID    & Amplitude & 3          & 16         & 16         & 0.005         & 10  \\
        \bottomrule
    \end{tabular}
 \end{adjustbox}
\end{table}

At last we want to point out how the depolarizing mechanism was constructed in our study

\subsection{Implementation of the SDP}

\subsubsection{Complexity analysis}
In this study, we use MOSEK in combination with CVXpy~\cite{Diamond2016} with an interior-point method, scaling its runtime with $\mathcal{O}(v^{3.5})$~\cite{Andersen2009TheHA} to solve the SDP, where $v$ represents the number of variables in the SDP. In our work, we embed a Choi matrix of size $d^2 \times d^2$, where $d = 2^n$ describes the dimensionality of the Hilbert space. This indicates that the number of variables scales biquadratically with the dimensionality of the Hilbert space of the used quantum classifier, and hence exponentially in terms of the number of qubits:
\begin{equation}
    v = d^4 = (2^n)^4
\end{equation}
With $d$ the dimensionality of the hilberspace and $n$ the number of qubits.
The biquadratic scaling also indicates why, in contrast to other methods like \citet{Guan_2021}, we can only exhibit our study on maximally three qubits.  
\medskip

\subsubsection{Numerical Considerations}
As already pointed out, we used CVXpy~\cite{Diamond2016} to implement the SDP in \autoref{eq:pre_order_sdp} combined with MOSEK as a solver. During optimization, we resolved multiple numerical issues. Initially, we approximated the constraint: $\text{Tr}_2 \left[ \mathcal{J} \right] = I_A$, by replacing it with: $\left\| \text{Tr}_2 \left[ \mathcal{J} \right] - I_A \right\|_\text{F} \leq \delta$, setting $\delta = \num{1e-10}$. Secondly, instead of explicitly stating the constraint that $\mathcal{J} \succeq 0$, we asserted on the variable itself that $\mathcal{J}$ is PSD, which typically guided the optimizer to feasible solutions. Lastly, we emphasize that our optimal noise channel was constructed on the test set, on which we further applied the attack. The main reasoning behind this was to provide the actual best possible noise channel when averaging over the test set.

\section{Overlapping Robustness in Optimal Noise Channels}\label{apx:coincideness_of_optimal_channels}
During our work, we found that typically the optimally constructed noise channels coincide in terms of their adversarial accuracy. This provides a challenging perspective as $\alpha$ and $\gamma$ seemingly control the contractiveness and the minimal possible measurement, respectively. Nevertheless, we want to emphasize that $\alpha$ provides an upper bound on the contractiveness of the noise channel by \autoref{def:alpha_gamma_channel}, and $\gamma$ \textit{a lower bound} on the minimal measurements. Those bounds suffice to control $\epsilon$-DP. Nevertheless, the optimizer might choose a different specific values describing the contraction and smallest possible measurement. We hence saw as outlined  for Iris Amplitude in \autoref{tab:adv_accs_opt}, that the adversarial accuracies typically where overlapping.
We hence simply used one value for the effect of $\alpha$ and $\gamma$ in the following when describing adversarial accuracy's.
\begin{table}[]
    \centering
    \caption{Adversarial Accuracies for Iris Amplitude, among various values of $\alpha$ and $\gamma$ parameterized by different values of depolarizing noise. Produced by the optimal noise channel}
    \label{tab:adv_accs_opt}
    \begin{adjustbox}{max width=\linewidth}
        
        \begin{tabular}{ll|c|c|c|c|c|c|c|c|c|c}
        \toprule
        \multicolumn{2}{l|}{}         & \multicolumn{10}{ c }{Adversarial Budget ($\epsilon$)} \\ 
        \midrule
         $\alpha$ & $\gamma$ & 0.0 & 0.03 & 0.05 & 0.08 & 0.1 & 0.13 & 0.15 & 0.18 & 0.21 & 0.23 \\
        \midrule
            1.0 & 0.0 & 1.0 & 1.0 & 1.0 & 1.0 & 1.0 & 0.9 & 0.5 & 0.17 & 0.05 & 0.03\\
            0.9 & 0.05 & 1.0 & 1.0 & 1.0 & 1.0 & 1.0 & 0.9 & 0.5 & 0.17 & 0.05 & 0.03\\
            0.7 & 0.15 & 1.0 & 1.0 & 1.0 & 1.0 & 1.0 & 0.9 & 0.5 & 0.17 & 0.05 & 0.03\\
            0.5 & 0.25 & 1.0 & 1.0 & 1.0 & 1.0 & 1.0 & 0.9 & 0.5 & 0.17 & 0.05 & 0.03\\
            0.1 & 0.45 & 1.0 & 1.0 & 1.0 & 1.0 & 1.0 & 0.9 & 0.5 & 0.17 & 0.05 & 0.03\\
            0.05 & 0.48 & 1.0 & 1.0 & 1.0 & 1.0 & 1.0 & 0.9 & 0.5 & 0.17 & 0.05 & 0.03\\
            0.0 & 0.5 & 0.43 & 0.43 & 0.43 & 0.43 & 0.43 & 0.43 & 0.43 & 0.43 & 0.43 & 0.43\\
        \bottomrule
    \end{tabular}
    \end{adjustbox}
\end{table}

\clearpage
\pagebreak
\newpage
\bibliographystyle{plainnat}
\bibliography{references}

\end{document}